\documentclass[conference]{IEEEtran}
\IEEEoverridecommandlockouts

\usepackage{fancyhdr, epsfig, epsf, amsthm, amsmath, amssymb, amsfonts, subfigure, color}
\usepackage{threeparttable}
\usepackage[noadjust]{cite}
\usepackage{dsfont}
\usepackage{enumerate}
\usepackage{comment}
\usepackage{soul}
\usepackage{kotex}
\usepackage{algorithm}
\usepackage{algpseudocode}
\usepackage[thinc]{esdiff}
\usepackage{multirow}
\usepackage{makecell}
\DeclareMathOperator*{\argmax}{argmax}
\newtheorem{theorem}{Theorem}

\newtheorem{proposition}{Proposition}
\newtheorem{remark}{Remark}

\newtheorem{assumption}{Assumption}

\def\({\left(}
\def\){\right)}

\def\b0{{\mathbf{0}}}

\includecomment{comment}	

\def\BibTeX{{\rm B\kern-.05em{\sc i\kern-.025em b}\kern-.08em
    T\kern-.1667em\lower.7ex\hbox{E}\kern-.125emX}}

\makeatletter
\newcommand{\linebreakand}{%
  \end{@IEEEauthorhalign}
  \hfill\mbox{}\par
  \mbox{}\hfill\begin{@IEEEauthorhalign}
}
\makeatother

\begin{document}
\setlength{\textfloatsep}{0pt}
\setlength{\belowdisplayskip}{1pt}
\title{\LARGE Over-the-Air Consensus for Distributed Vehicle Platooning Control
}
\author{\IEEEauthorblockN{Jihoon Lee$^*$, Yonghoon Jang$^*$, Hansol Kim$^\ddagger$,  Seong-Lyun Kim$^*$, Seung-Woo Ko$^\S$}\\
\IEEEauthorblockA{$^*$School of EEE, {Yonsei University}, Seoul, Korea, email: \{jhlee, slkim\}@ramo.yonsei.ac.kr, matice1116@yonsei.ac.kr}
\IEEEauthorblockA{$^\ddagger$Dept. of Cont. and Inst. Eng., {Korea Maritime and Ocean University}, Busan, Korea, email: hansol@kmou.ac.kr}
\IEEEauthorblockA{$^\S$Dept. of Smart Mobility Eng., {Inha University}, Incheon, Korea, email: swko@inha.ac.kr}
}
\maketitle

\begin{abstract}
A distributed control of vehicle platooning is referred to as \emph{distributed consensus} (DC) since many \emph{autonomous vehicles} (AVs) reach a consensus to move as one body with the same velocity and inter-distance. For DC control to be stable, other AVs' real-time position information should be inputted to each AV's controller via \emph{vehicle-to-vehicle} (V2V) communications. On the other hand, too many V2V links should be simultaneously established and frequently retrained, causing frequent packet loss and longer communication latency. We propose a novel DC algorithm called \emph{over-the-air consensus} (AirCons), a joint communication-and-control design with two key features to overcome the above limitations. First, exploiting a wireless signal's superposition and broadcasting properties renders all AVs' signals to converge to a specific value proportional to participating AVs' average position without individual V2V channel information. Second, the estimated average position is used to control each AV's dynamics instead of each AV's individual position. Through analytic and numerical studies, the effectiveness of the proposed AirCons designed on the state-of-the-art New Radio architecture is verified by showing a $14.22\%$ control gain  compared to the benchmark without the average position.
\end{abstract}
\section{Introduction}
\emph{Vehicle platooning} (VP) is one cooperative driving technique,  where many \emph{autonomous vehicles} (AVs) move as one body with the same velocity while maintaining their inter-distances \cite{ACC_1}. VP has various potentials, such as fuel savings and traffic control (see, e.g., \cite{energy_efficient_2}). There have been extensive studies in academia and industry to make VP fully automated without a central controller. Such a distributed platoon control is an analogy to reaching a consensus among multiple agents, referred to as a \emph{distributed consensus} (DC)~\cite{distibuted_consensus}.
 
For DC to be effectively designed, one essential requirement is not to make the disturbances of one AV's position propagate along with the platoon, referred to as \emph{string stability} \cite{String_stability1}. To this end, it is required to obtain many AVs' kinetic information, including \emph{non-line-of-sight} (NLoS) AVs. It can be viable via the recent rise of  \emph{vehicle-to-vehicle} (V2V) communications by transmitting a modulated signal embedding each AV's own kinetic information to AVs located in NLoS \cite{NLoS_V2V}. On the other hand, practical factors of V2V channels cause unpredictable latency and packet loss. The main research thrust is thus to reach string stability under those considerations. The effect of communication latency is analyzed in \cite{impact_of_time_delay} that each AV's asynchronous control due to different latency makes VP string unstable. To maintain string stability, synchronizing their control timings is proposed under the assumption that the latency of each V2V link is equivalent. In \cite{VP_delay2}, the stochastic distribution of the resultant latency is derived based on modeling a packet as a Bernoulli trial, which is used to determine several control parameters to achieve string stability in the presence of packet loss. 

The above prior works do not consider the issue of training \emph{channel state information} (CSI), which is essential yet challenging in practical scenarios due to the following reasons. First, an AV's high velocity renders the relevant V2V CSIs' coherence times shorter, requiring frequent channel retraining before out-of-date. Second, each AV needs to share its information with multiple AVs nearby. In other words, a large number of V2V links should be established simultaneously using different pilots. A few pilots can be non-orthogonal, causing severe pilot contamination that brings about more frequent packet losses and longer communication latency \cite{pilot_contamination}. 

To address the above challenges, we design a consensus algorithm for DC control called \emph{over-the-air consensus} (AirCons), which achieves a consensus among multiple APs in the air without individual V2V CSI training. It is the integrated architecture of communication and control with two key features. First, each AV iteratively exchanges the signal embedding its position information, converging to a specific value proportional to the transmitting AVs' average position. Second, each AV utilizes the estimated average position to adjust its dynamics, which helps achieve string stability. AirCons is different from the over-the-air computation that requires individual CSIs to make all signals coherently combined (see, e.g., \cite{AFC_imperfect_CSI}). 
AirCons does not rely on AVs' individual position information. Instead, a superimposed signal is normalized by the sum of CSIs obtained via common pilot transmissions, thereby reducing the burden of training individual V2V CSIs. Besides, this operation enables the convergence to the targeted value closely, which is proved using random matrix theory. We design AirCons based on \emph{New Radio} (NR) architecture, verified via extensive simulation that AirCons can satisfy string stability with a significant gain on platooning error compared to the benchmark without the average information.           
\section{System Model}\label{Section:SystemModel}
This section introduces our system model, including network and signal models, and a distributed platoon control.  
\begin{figure}[t]
\centering
\centering
\includegraphics[width=8.8cm]{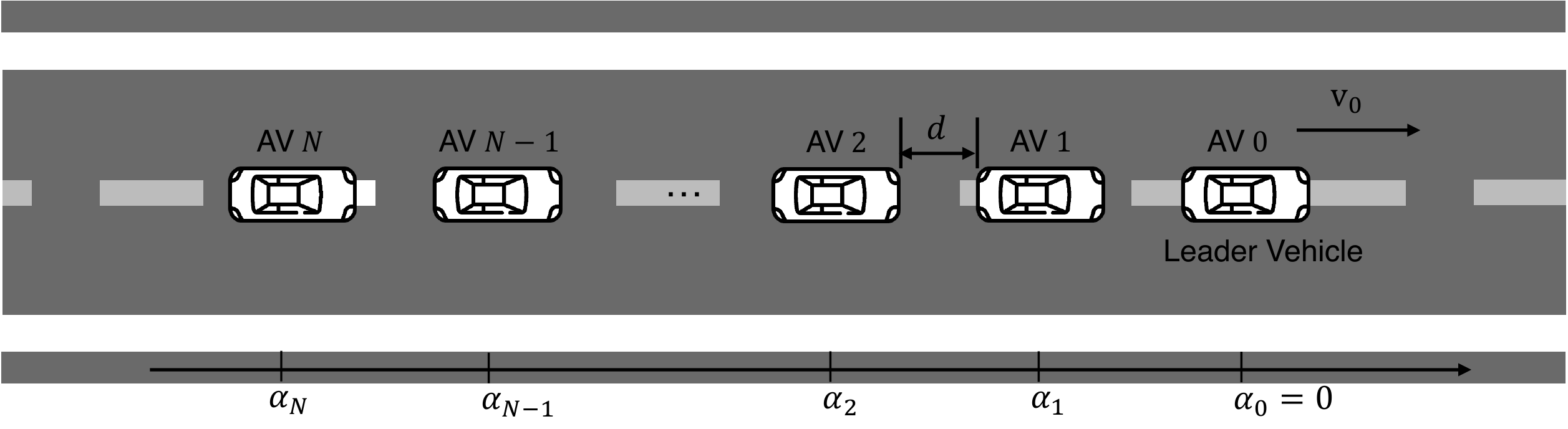}
\caption{A graphical illustration of platooning control with $(N+1)$ AVs on a straight highway.}\vspace{4pt}\label{Platoon}
\end{figure}
\subsection{Network Model}\label{Subsec:ServiceCaching}
We consider the scenario with $(N+1)$ AVs located on a uni-directional and single-lane highway, as illustrated in Fig. \ref{Platoon}. These AVs aim to form a platoon to move from left to right like one body. The rightmost AV, say AV $0$, is a leader determining the platoon's speed, denoted by ${v}_0$. The remaining AVs are followers, say AV $n\in \{1,...N\}$, attempting to keep their velocities and inter-distances between adjacent AVs at $v_0$ and a predetermined target distance $d$, respectively. To this end, each AV should share its location with the neighboring AVs, helping control their dynamics to stabilize the platoon. The detailed information sharing and platooning control mechanisms are introduced in the sequel. 

\subsection{Signal Model}
We consider that every AV is equipped with a full-duplex antenna, enabling simultaneous V2V transmission and reception among multiple AVs \cite{overtheair1001}.
\emph{Orthogonal Frequency-Division Multiplexing} (OFDM) is used as a transmit waveform, where a wide-band channel is divided into multiple narrow-band orthogonal sub-carriers. The number of sub-carriers is given as $F$. 
For sub-channel $f\in\{1,\cdots,F\}$ at time $t$, the channel coefficient from AVs $m$ to $\ell$ is denoted by $h_{\ell,m}^{(f)}(t)$, which is expressed as a product between path-loss and short term fading. Following the common assumption of Rayleigh fading, the channel $h_{\ell,m}^{(f)}(t)$ is modeled as a complex Gaussian random variable $\mathcal{CN}(0,d_{\ell,m}^{-\frac{\eta}{2}})$, where $d_{\ell,m}$ and $\eta$ are the inter-distance between AVs $\ell$ and $m$, and path-loss exponent, respectively. 

We denote $x_m^{(f)}(t)$ AV $m$'s baseband signal for sub-channel $f$ at time $t$, embedding its position information under its power constraint explained in the sequel. Assuming that AVs in $\mathbb{N}^{(f)}$ is granted to transmit their signals using sub-channel $f$, AV $\ell$'s received signal for sub-channel $f$ at time $t$, denoted by $y_{\ell}^{(f)}(t)$, is given as       
\begin{align}
    y_{\ell}^{(f)}(t)=\sum_{m \in \mathbb{N}^{(f)}\setminus\{\ell\}}h^{(f)}_{\ell,m}(t)x_m(t)+z_{\ell}^{(f)}[t],\label{Received_Signal}
\end{align}
where $z_{\ell}^{(f)}[t]$ represents an i.i.d. \emph{additive white Gaussian noise} (AWGN) following $\mathcal{CN}(0,\sigma_n^2)$. Note that AV $\ell$'s transmit signal is not included in \eqref{Received_Signal}, assumed to be perfectly canceled out using a self-interference cancellation technique.
\vspace{-5pt}
\subsection{Distributed Platooning Control}\label{Subsec:{Design of Controller for Platooning}}
\vspace{-5pt}
We explain a distributed platooning control mechanism based on the assumptions below. 
\begin{assumption}[Given State Information]\label{Assumption1} 
\emph{All AVs know the following two state information. 
\begin{itemize}
    \item \textbf{Own state information}: Each AV knows its own state information, say velocity $v_n$ and absolute location $\mathbf{p}_n$ for AV $n$, which are obtainable using various positioning techniques, e.g., global positioning system and cellular positioning system \cite{1003}. 
    \item \textbf{Leader's state information}: A dedicated channel is granted to AV $0$. Then, its velocity $v_0$ and $\mathbf{p}_0$ are periodically broadcast to the other AVs. 
\end{itemize}
} 
\end{assumption}

By Assumption \ref{Assumption1}, each AV can calculate the relative state information concerning the leader AV, given as
\begin{align}
    \alpha_n=\|\mathbf{p}_n-\mathbf{p}_0\|, \quad \beta_n=v_n-v_0,
\end{align}
where $\|\mathbf{x}\|$ represents the Euclidean distance of $\mathbf{x}$. Given $\{\alpha_n, \beta_n\}$, the aforementioned criteria for stabilizing the platoon are rewritten in terms of $\alpha_n$ and $\beta_n$ as follows:
\begin{align}
    & \alpha_{n} \xrightarrow{}n d,\quad \beta_{n} \xrightarrow{}0. \label{Eq: Goal of the Platoon}
\end{align}
Here, the desired inter-vehicle spacing $d$ is specified in Sec. \ref{Subsec:ServiceCaching}. To this end, we follow a second-order distributed platooning control, which is widely used in literature (see e.g.,  \cite{typical_controller1004} and \cite{typical_controller1005}). Specifically, we define a vector $\mathbf{c}_n=[c_{n,1}, c_{n,2}, \cdots, c_{n,N},]^T$, whose element $c_{n,m}=1$ if AV $m$'s state information obtained via V2V communication is involved to determine AV $n$'s control and $c_{n,m}=0$ otherwise. Then, AV $n$ can control its accelerator, denoted by $\mu_n$, according to the following equation:
\begin{align}
    \mu_n&=\frac{1}{\|\mathbf{c}_n\|}\sum_{m=1}^N {\kappa c_{n,m}(\alpha_{n}-\alpha_m-(n-m)d)}-\delta\beta_n+\xi_n,\label{Eq: Typical Controller}
\end{align}
where $\kappa$ and $\delta$ represent stiffness and damping coefficients, respectively, and $\xi_n$ represents the control term using information obtained via the third-party sensors (e.g., a predecessor AV's position and velocity obtained via RADAR). Note that  $\mathbf{c}_n$, $\kappa$, and $\delta$ are parameters predetermined by the concerned control policy and assumed to be given advance. 

We focus on the terms relevant to other AVs' relative distances, say $\{\alpha_m\}$, which should be delivered through wireless links. To this end,  \eqref{Eq: Typical Controller} is rewritten as
\begin{align}
    \mu_n&=\kappa (\alpha_{n}-nd)-\delta\beta_n+
    \frac{\kappa d}{|\mathbb{S}_n|}\sum_{m\in \mathbb{S}_n} {m}+\xi_n-\frac{\kappa}{|\mathbb{S}_n|}\sum_{m\in \mathbb{S}_n}{\alpha_m},\label{Eq:Revised_controller}
\end{align}
where $\mathbb{S}_n=\{m|c_{n,m}=1\}$ and $|\cdot|$ is a cardinality operator. Two key observations are made as follows. First, every term except the last one can be computed using the state information mentioned in Assumption \ref{Assumption1}. Second, the last term is determined by the average of relative distances involved in AV $n$'s control, denoted by $\gamma_n$ as
\begin{align}
\gamma_n=\frac{1}{|\mathbb{S}_n|}\sum_{m\in \mathbb{S}_n }{\alpha_m}.\label{Definition_Gamma}
\end{align}
In other words, the information required to control AV $n$'s accelerator is the sum of AVs' relative distance $\gamma_n$, not individual $\{\alpha_m\}_{m\in \mathbb{S}_n}$.   
    
\section{Over-the-Air Consensus: Principle and Design}\label{Section: AirCons}
This section introduces AirCons, a novel algorithm acquiring the average of AVs' relative distances in the air by exploiting a wireless channel's broadcasting and superposition properties. First, we explain the principle and design of AirCons based on its two key features contrasting from conventional V2V communications. 
\subsection{Overview and Key Features}\label{Subsection:Overview}
AirCons is designed for each AV to acquire the average of neighbor and its own AVs' relative distances by capturing a superimposed signal transmitted from the neighbors. To explain, AVs keep exchanging their received signals until all of them converge to $\zeta_n$, which is referred to as \emph{consensus}. In results, \emph{consensus} $\zeta_n$ is one-to-one mapped to average of neighbor $\gamma_n$.
After reaching a consensus, each AV's accelerator is updated according to \eqref{Eq:Revised_controller}. The detailed process to reach a consensus will be explained in Sec. \ref{subsec:Alg_Description}.

We explain how AirCons works from the perspective of an entire platooning network. For stabilizing the platoon, each AV is required to participate in multiple processes to reach different consensuses simultaneously, including not only its one (e.g., $\gamma_n$ for AV $n$) but also others ($\{\gamma_m\}$ when $n\in \mathbb{S}_m$). The total number of V2V links relevant to AirCons is thus $\sum_{n=1}^N |\mathbb{S}_n|$. 
It is a heavy burden to establish all individual V2V links by training their CSIs using different pilots. On the other hand, we highlight the following two features of AirCons, deviated from the conventional V2V architecture.  
\subsubsection{Independent Consensus Process} The basic unit to which the sub-carrier is allocated should be a consensus process, not an individual V2V link. In other words, sub-carriers allocated to one consensus process are shared by all involved AVs, while the other AVs are not allowed to use them. As a result, each consensus process can work independently without interfering with the others. 
\subsubsection{Superimposed Channel Estimation} Each AV needs to know superimposed channel coefficients. Consider a consensus group $\mathbb{S}_n$ using sub-carrier $f$. For a consensus group $\mathbb{S}_n$, we denote transmitter set $\mathbb{T}_n=\mathbb{S}_n \cup \{n\}$. For AV $\ell$, the received superimposed signal is $\sum_{m \in  \mathbb{T}_n /\{\ell\}}h^{(f)}_{\ell,m}(t)$. To estimate it, a single pilot waveform is commonly used for AVs in $\mathbb{T}_n$, denoted by $w_n$. Assume high \emph{signal-to-noise ratio} (SNR).
When AVs in $\mathbb{T}_{n}$ transmits $w_{n}$ simultaneously, 
AV $\ell$ can figure out the superimposed channel coefficients from the received signal $\left\{\sum_{m \in \mathbb{T}_{n}/\{\ell\}}h^{(f)}_{\ell,m}(t)\right\}w_{n}$, according to \eqref{Received_Signal}. The estimated one plays a normalization factor for the consensus, which will be explained in the following~subsection.

\subsection{Algorithm Description}\label{subsec:Alg_Description}
In this section, we elaborate on AirCons' algorithm step-by-step on how AVs in $\mathbb{S}_n$ achieve a consensus as the average of their relative distances.
\subsubsection{Resource Block Configuration and Assumption}\label{Sec: Resource_block}
\begin{figure}[t]
\centering
\includegraphics[width=8.8cm]{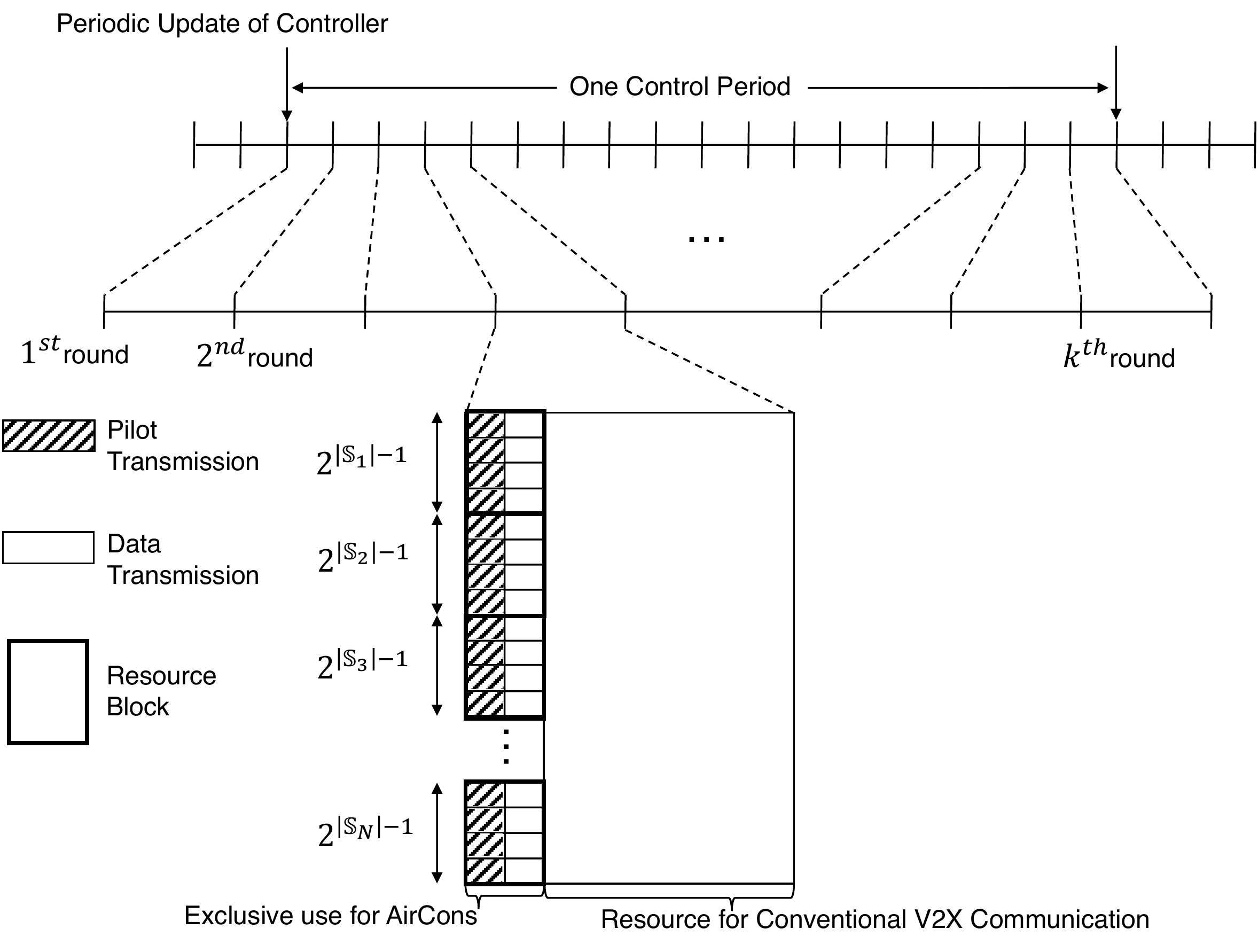}
\caption{A graphical illustration of resource block configuration in AirCons process with multiple rounds.}\vspace{4pt}\label{Algorithm_Timeline}
\end{figure}
A consensus is achieved through $K$ rounds of updating AVs' transmit signals. For one round update explained in the sequel, $2^{|\mathbb{S}_n|-1}$ contiguous sub-carriers and two adjacent OFDM symbols are needed, defined as a \emph{resource block} (RB), as illustrated in Fig. \ref{Algorithm_Timeline}. 
The set of the concerned sub-carriers are denoted by $\mathbb{F}_n=\{1,\cdots, 2^{|\mathbb{S}_n|-1}\}$. The assigned first and second OFDM symbols are used for pilot and data transmissions, respectively. 

Next, we assume that all channel coefficients in a RB are equivalent, namely,
\begin{align}\label{Channel_Assumption}
h_{\ell,m}^{(f_1)}(t_1)=h_{\ell,m}^{(f_2)}(t_2), 
\end{align}
for $f_1,f_2\in \mathbb{F}_m$ and 
$t_1,t_2\in[(k-1)T, (k-1)T+2\delta]$. 
Here, $\delta$ and $T$ are the durations of symbol and round respectively, and $k$ is the index of round.  The assumption will be well justified in the following remark.
\begin{remark}[Coherence Time and Bandwidth]\label{Coherence_Time_Bandwidth}\emph{It is recommended for NR V2V communications to use a short symbol duration, e.g., $\delta=16.7$ ($\mu$sec) when using the center frequency of $f_c=5.9$~(GHz) according to flexible numerology \cite{Numerology}. Denoting a size of consensus group $S=|\mathbb{S}_n|+1$, the resulting time and frequency ranges of one RB for a consensus process are $2\delta=33.4$ ($\mu$sec) and $\frac{2^{S-2}}{\delta}=60 \cdot {2^{S-2}}$ (kHz), which will be compared with the following analysis of coherence time and bandwidth to justify~\eqref{Channel_Assumption}.   
\begin{itemize}
\item \textbf{Coherence Time}: Given the relative velocity $v=200$ (km/h), the resulting coherence time $T_c$ is  $T_c=\frac{c}{f_c v}=915$ ($\mu$sec) with the light speed $c$, which is larger than one RB's time coverage~$2\delta$.  
\item \textbf{Coherence bandwidth}: NR specifies V2V channel's delay spread ranging from $5$ to $56$ (nsec) \cite{3gpp.37.885}. The corresponding coherence bandwidth $B_c$, inversely proportional to the delay spread, is at least $17.8$ (MHz). In other words, the flat fading assumption of \eqref{Channel_Assumption} makes sense when AVs for one consensus is less than $10$, i.e.,~$S\leq 10$. 
\end{itemize}
}
\end{remark}\label{RM: Coherence_time&bandwidth}
Hereafter, we focus on a typical consensus group, and 
all channel coefficients relevant to the typical consensus group at the $k$-th round is unified as $h_{\ell,m}[k]$, where the index of sub-carrier $f$ is omitted for ease of notation. 
Last, we utilize only an in-phase term for AirCons, defined as
\vspace{-5pt}
\begin{align}\label{a_k}
a_{\ell,m}[k]=\mathsf{Re}\{h_{\ell,m}[k]\},
\end{align}
whereas a quadrature term is reserved for future extension.  

\subsubsection{Pilot Encoding \& Transmission} 
Consider the OFDM symbol assigned for pilot transmissions, which will be used to find the sum of in-phase channels' magnitudes, given as
\begin{align}
     \sum_{m \in \mathbb{T}_{n}/\{\ell\}}|a_{\ell,m}[k]|=\sum_{m \in \mathbb{T}_{n}/\{\ell\}} \mathsf{b}\left(a_{\ell,m}[k]\right)\cdot a_{\ell,m}[k],
\end{align}
where $\mathsf{b}(x)=1$ if $x\geq0$ and $\mathsf{b}(x)=-1$ otherwise. Under the condition without knowledge of individual channel coefficients, AV $\ell$ attempts to check all possible binary combinations of $\{\mathsf{b}(a_{\ell,m}[k])\}_{m\in \mathbb{T}_{n}/\{\ell\}}$ and find the maximum one, namely,
\begin{align}\label{Optimization_1}\tag{P1}
\left\{\mathsf{b}(a_{\ell,m}[k])\right\}&_{m\in \mathbb{T}_{n}/\{\ell\}}=\{I_m^*[k]\}_{m\in\mathbb{T}_{n}/\{\ell\}}\\&=\argmax_{\{I_m\}\in\{-1,1\}^{|\mathbb{S}_{n}|}} \sum_{m \in \mathbb{T}_{n}/\{\ell\}} I_m \cdot a_{\ell,m}[k].\nonumber
\end{align}
Problem \ref{Optimization_1} is implementable by encoding the aforementioned pilot waveform $w_n$ differently for each combination. As an example of one sequence $I_m$, AV $m$ in $\mathbb{T}_n$ encodes $w_n$ by multiplying the corresponding indicator $I_m$. AV $\ell$'s received signal is given as $\left(\sum_{m \in \mathbb{T}_{n}/\{\ell\}}h_{\ell,m}[k] I_m \right)\sqrt{{P}}w_{n}+e_{\ell}[k]$, where ${P}$ is the transmit power budget and $e_{\ell}[k]$ represents the AWGN when receiving a pilot waveform. Without loss of generality, the waveform's energy $|w_{n}|$ is fixed to one. By multiplying the conjugate of $w_{n}$, say ${\bar{w}_{n}}$ where $\bar{x}$ denotes the conjugate of $x$, it is coherently demodulated as
\begin{align}
    r_{\ell}[k]=\sqrt{P} \left(\sum_{m \in \mathbb{T}_{n}/\{\ell\}}I_m \cdot h_{\ell,m}[k]\right) +\tilde{e}_{\ell}[k], \label{Eq: Waveform_Encoding}
\end{align}
where $\tilde{e}_{\ell}[k]={e_{\ell}[k]\bar{w}_{n}}$. When the signal strength is strong enough to ignore the noise term, the real part of \eqref{Eq: Waveform_Encoding} becomes proportional to the corresponding argument in \ref{Optimization_1}. 

Each binary combination is one-to-one mapped into a different sub-carrier in RB. Given $|\mathbb{S}_n|$ AVs, there exist $2^{|\mathbb{S}_n|}$ binary combinations, which can be reduced in half when considering reversal counterparts. As a result, the number of sub-carriers required for implementing \ref{Optimization_1} is  $\frac{2^{|\mathbb{S}_n|}}{2}=2^{|\mathbb{S}_n|-1}$, equivalent to the number of sub-carriers in one RB. We can find the optimal binary combination of \ref{Optimization_1} by choosing the largest received signal strength. The resulting demodulated signal of the $k$-th round are denoted by $r_{\ell}^*[k]$ as
\begin{align}
    r_{\ell}^*[k]=\sqrt{P} \left(\sum_{m \in \mathbb{T}_{n}/\{\ell\}}I_m^*[k] \cdot h_{\ell,m}[k]\right) +\tilde{e}_{\ell}[k].
\end{align}
\begin{figure}[t]
\centering
\includegraphics[width=8.8cm]{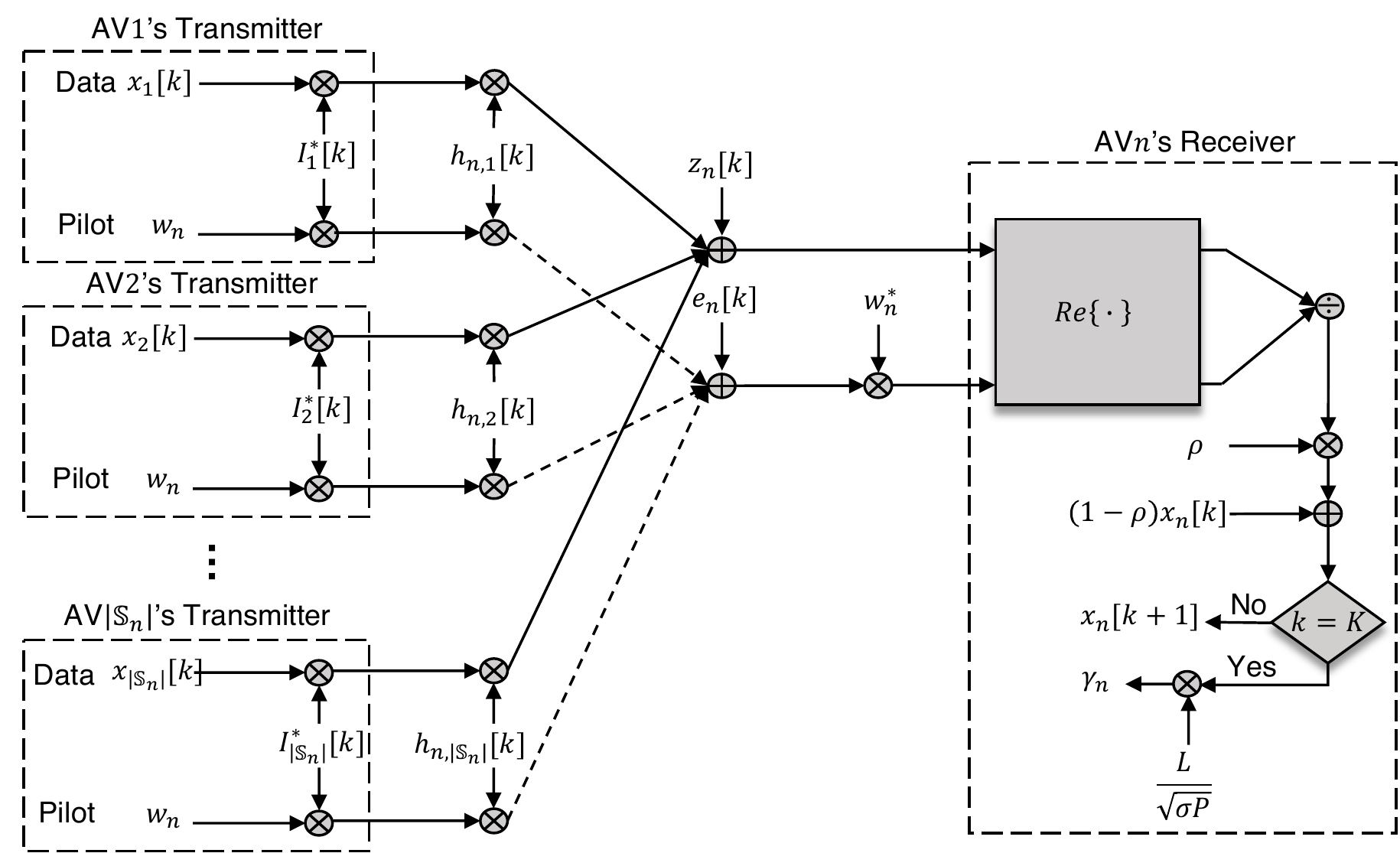}
\caption{A graphical example of AirCons in step showing how optimal $I_m^*$ is chosen and used for AirCons. }\vspace{4pt}\label{Schematic_Diagram}
\end{figure}
\subsubsection{Data Encoding} 
Each AV encodes its signal embedding its relative distance.
Consider AV $m$ in $\mathbb{T}_n$, whose relative distance is $\alpha_m$. We assume that $\alpha_m$ stays constant during one consensus process, while it may be changed over different processes. We adopt an \emph{amplitude modulation} (AM) to encode $\alpha_m$ into an initial base-band signal $s_{m}$, namely,  
\begin{align}
    s_{m}= \frac{\sqrt{\sigma P}\alpha_m}{L},
\end{align}
where $\sigma$ is a power scaling factor set as $1$ unless specified. Given $(N+1)$ AVs and $d$ desired inter-distance, the normalization factor $L$ is set as the desired maximum platoon length $L=Nd+\Delta$, where $\Delta$ is the control error margin. 

\subsubsection{Data Transmission \& Reception}
In the OFDM symbol duration assigned for the $(k+1)$-th round's data transmission, each AV in $\mathbb{T}_n$ transmits its signal by mixing the previous transmitted one with the term relevant to the received signals in the $k$-th round. 
Specifically, AV $m$ transmits its signal $I^*_{m}[k]\cdot x_{m}[k]$, where $x_{m}[k]$ is defined in a recursive form as 
\begin{align}\label{Eq:Updating_x}
    x_{m}[k+1]=
    \left\{
    \begin{aligned}
    &s_{m},  &\textrm{if $k=0$,}\\
    &(1-\rho) x_{m}[k]+ \rho \left(\frac{\mathsf{Re}\{y_{m}^{*}[k]\}}{\mathsf{Re}\left\{r_{m}^*[k]\right\}}\right),  &\textrm{if $k>0$.}\\
    \end{aligned}\right.
\end{align}
Here, $0<\rho<1$ is a weighted factor of the new observations that is the ratio of in-phases between the received data signal and demodulated pilot signals specified in \eqref{Eq: Waveform_Encoding}, given as
\begin{align}
  y_{m}^{*}[k]&=\sum_{\ell\in \mathbb{T}_{n}/\{m\}} h_{m,\ell}[k]I_\ell^*[k] x_{\ell}[k]+z_{m}[k],\nonumber\\
    r_{m}^{*}[k]&=\sum_{\ell\in \mathbb{T}_{n}/\{m\}} h_{m,\ell}  [k] I_\ell^*[k]\sqrt{P}+\tilde{e}_{m}[k].\nonumber
\end{align}

The schematic diagram including each step of AirCons is graphically illustrated in Fig. \ref{Schematic_Diagram}. In the following theorem, we will explain the asymptotic result of \eqref{Eq:Updating_x}, which is the main principle of AirCons.   

\begin{theorem}[Consensus]\label{Theorem1}\emph{Consider a high \emph{signal-to-noise ratio} (SNR) regime.  Given $0<\rho<1$, all AVs' transmitting signals, say $x_m[k]$ of \eqref{Eq:Updating_x} for $m\in\mathbb{T}_{n}$, always reach the consensus  to a value proportional to weighted average of the relative distances $\alpha_m$ as $k$ increases, namely,
\begin{align}
    \lim_{k\rightarrow \infty} x_m[k]=\frac{\sqrt{\rho P}}{L}\sum_{m\in \mathbb{T}_n}w_m \alpha_m,\quad \forall m\in\mathbb{T}_n,
\end{align}
where $\mathbf{\omega}=[\omega_1,\cdots, \omega_{S}]$ is an $S$-by-$1$ column vector  satisfying $\mathbf{1}^T\mathbf{\omega}=1$ and $\mathbf{\omega} \geq 0$.
}
\end{theorem}
\begin{proof}
See Appendix \ref{Appendix_A}.
\end{proof}

\subsubsection{Data Decoding}
AV $n$'s transmitting signal after the consensus, say $x_n[K]$ specified in Theorem \ref{Theorem1}, which can be converted into the approximated relative distance sum by dividing $\frac{L}{\sqrt{\rho P}}$, namely, 
\begin{align}\label{Eq: Data_Decoding}
    \tilde{\zeta}_n= \frac{\sqrt{\sigma P}}{L}\left(\lim_{k\rightarrow \infty} x_n[k]\right)=\sum_{m\in\mathbb{T}_n}\omega_m \alpha_m.
\end{align}
Since $\tilde{\zeta}_n$ includes its own relative position $\alpha_n$, we calculate average of neighbors $\gamma_n$ by excluding $\alpha_n$. That is, $\gamma_n=\frac{\tilde{\zeta_n}(|\mathbb{S}_n|+1)-\alpha_n}{|\mathbb{S}_n|}$, and the result is inputted into the controller specified in \eqref{Eq:Revised_controller}, enabling AV $n$ to adjust its accelerator $\mu_n$ on time. The deviation from the ground truth $\zeta_n$ will be analyzed in the following subsection. 

\begin{figure}[t] 
\centering
\centering
\includegraphics[width=8.8cm]{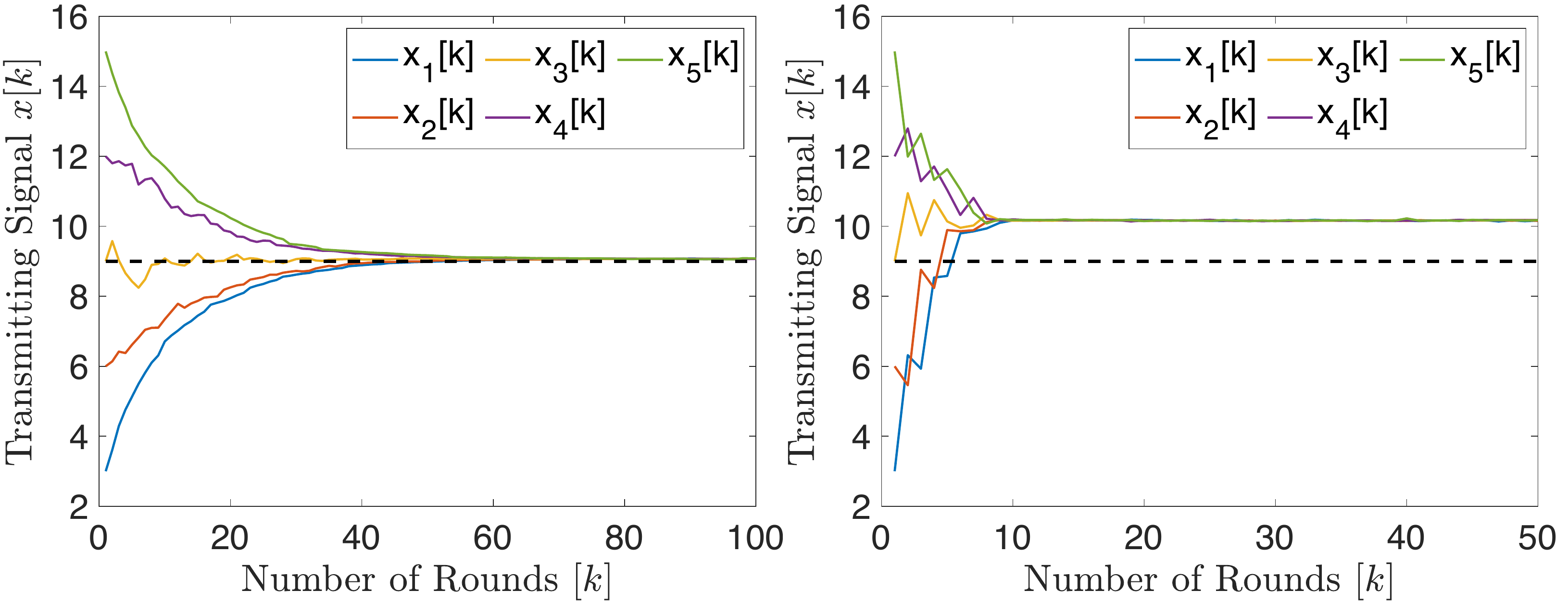}
\caption{The trajectory of consensus processes with $\rho=0.2$ (left) and $\rho=0.9$ (right). Dotted line represents ground truth average. The other parameters are specified in Sec. \ref{Section: Numerical_Results}. }\vspace{4pt}\label{Fig_AirCons}
\end{figure}

\begin{remark}[Effect of $\rho$]\emph{The weighted factor $\rho$ controls the trade-off between convergence speed and the accuracy of the consensus, as shown in Fig. \ref{Fig_AirCons}. A small $\rho$ updates the signal in a conservative manner by giving more portion to the current value. Eventually, the final consensus result is more close to the ground-truth linear average with slower convergence speed. As $\rho$ increases, on the other hand, the new observation is more involved in the signal, leading to a faster convergence with a certain error compared with the ground truth. Through extensive numerical studies, we set $\rho=0.9$, to guarantee a fast convergence with an acceptance accuracy.}
\end{remark}\label{Remark2}

\subsection{Deviation Analysis}
Denote $\epsilon_n$ the deviation between the consensus value and the ground truth, namely,
\begin{align}
    \epsilon_n=\tilde{\zeta}_n-{\zeta}_n,
\end{align}
which depends on the realizations of relevant channels, say $\{|a_{m,\ell}[k]|\}$ for all $k$. Following a similar approach in \cite{Average_Consensus_Controller1006} leads to expressing $\epsilon_n$ in closed~form as
\begin{align}
    \epsilon_n=\frac{\rho}{S(S-1)}\lim_{z\rightarrow 1} \left[\sum_{n,m \in \mathbb{T}_n}\left(\sum_{k=0}^\infty v_{n,m}[k]z^{-k}\right)\right],\nonumber
\end{align}
where $v_{n,m}[k]=\frac{\sum_{i \in \mathbb{T}_n/\{n\}}\left(|a_{n,m}[k]|-|a_{n,i}[k]|\right)}{\sum_{i \in \mathbb{T}_n/\{n\}}|a_{n,i}[k]|}x_{m}[k]$ if $n \neq m$ and $0$ otherwise. Its expectation over a random sequence of $\{v_{n,m}[k]\}$, say $\mathsf{E}[\epsilon_n]$, is given as
\begin{align}
    \mathsf{E}[\epsilon_n]
    =&\frac{\rho}{S(S-1)}\mathsf{E}\left[\lim_{z\rightarrow 1}\sum_{n,m \in \mathbb{T}_n} \left(\sum_{k=0}^{\infty}v_{n,m}[k])z^{-k}\right)\right]\nonumber\\=&\frac{\rho}{S(S-1)}\lim_{z\rightarrow 1}\sum_{n,m \in \mathbb{T}_n} \left(\sum_{k=0}^{\infty}\mathsf{E}[v_{n,m}[k])]z^{-k}\right).\nonumber
\end{align}
Here, the term $\mathsf{E}[v_{n,m}[k]]$ is lower bounded as
\begin{align}
   \mathsf{E}\left[v_{n,m}(k)\right]
   &=\sum_{i \in \mathbb{T}_n/\{n\}}\mathsf{E}\left[\frac{|a_{n,m}[k]|}{|a_{n,i}[k]|}-1\right]x_{m}[k]\nonumber
   \\&\overset{(a)}{\geq}\sum_{i \in \mathbb{T}_n/\{n\}}\left(\frac{\mathsf{E}\left[a_{n,m}[k]\right]}{\mathsf{E}\left[a_{n,i}[k]\right]}-1\right)x_{m}[k],\label{Expectation_v}
\end{align}
where (a) follows from Jensen's inequality. The equality condition of \eqref{Expectation_v} is that all channels are time-invariant and fixed as their expectations defined as ${\eta}_{\ell,m}=\mathsf{E}\left[|a_{\ell,m}[k]|\right]$, and $\tilde{\zeta}_n$ is reduced to $\boldsymbol{v}\boldsymbol{\alpha}$, where $\boldsymbol{\alpha}=[\alpha_1,\cdots,\alpha_{S}]^T$ and $\boldsymbol{v}$ is a left eigenvalue of the following matrix $\boldsymbol{U}\in\mathbb{R}^{S \times S}$:
\begin{align}
    \boldsymbol{U}=
  &\begin{pmatrix}
        1-\rho & \frac{\eta_{1,2}}{\sum_{m \in \mathbb{S}_{n,1}}\eta_{1,m} } &\dots &\frac{\eta_{1,S}}{{\sum_{m \in \mathbb{S}_{n,1}}\eta_{1,m} }}\\ \frac{\eta_{2,1}}{\sum_{m \in \mathbb{S}_{n,2}}\eta_{2,m} }& 1-\rho & \dots & \frac{\eta_{2,S}}{\sum_{m \in \mathbb{S}_{n,2}}\eta_{2,m}} \\ \vdots & \vdots &\ddots &\vdots\\\frac{\eta_{S,1}}{\sum_{m \in \mathbb{S}_{n,S}}\eta_{S,m} }& \frac{\eta_{S,2}}{\sum_{m \in \mathbb{S}_{n,S}}\eta_{S,m}}&\dots &1-\rho
  \end{pmatrix},\nonumber
\end{align}
In other words, $\mathsf{E}[\epsilon_n]$ is lower bounded as
\begin{align}
\mathsf{E}\left[\epsilon_n\right]\geq \boldsymbol{v}\boldsymbol{\alpha}-\zeta_n.\label{Eq:Epsilon_lower_bound}
\end{align}
In the following proposition, we confirm that the above lower bound can be simplified when the inter-distance requirement in \eqref{Eq: Goal of the Platoon} is satisfied.     
\begin{proposition}[Consensus Deviation]\label{Proposition1}\emph{Assume that the inter-distance between adjacent AVs is equivalent, e.g., $d_{m-1,m}=d_{m,m+1}$ for all $m\in \mathbb{T}_n$. Then, the expectation of the consensus deviation $\epsilon_n$ becomes non-negative, e.g., $\mathsf{E}[\epsilon_n]\geq 0$.}
\end{proposition}
\begin{proof}
See Appendix \ref{Appendix_B}.
\end{proof}
Through extensive simulations, the tightness of result in Proposition \ref{Proposition1} is verified enough to ignore the gap between the two. Besides, the result works well even when there exists a certain level of control error on inter-distance. As a result, AirCons is verified to reach the consensus to the accurate average of relative distances under the condition of equivalent inter-distance with an acceptance relaxation.   
\section{Numerical Results}\label{Section: Numerical_Results}

\begin{figure*}[t] 
\centering
\centering
\includegraphics[width=18cm]{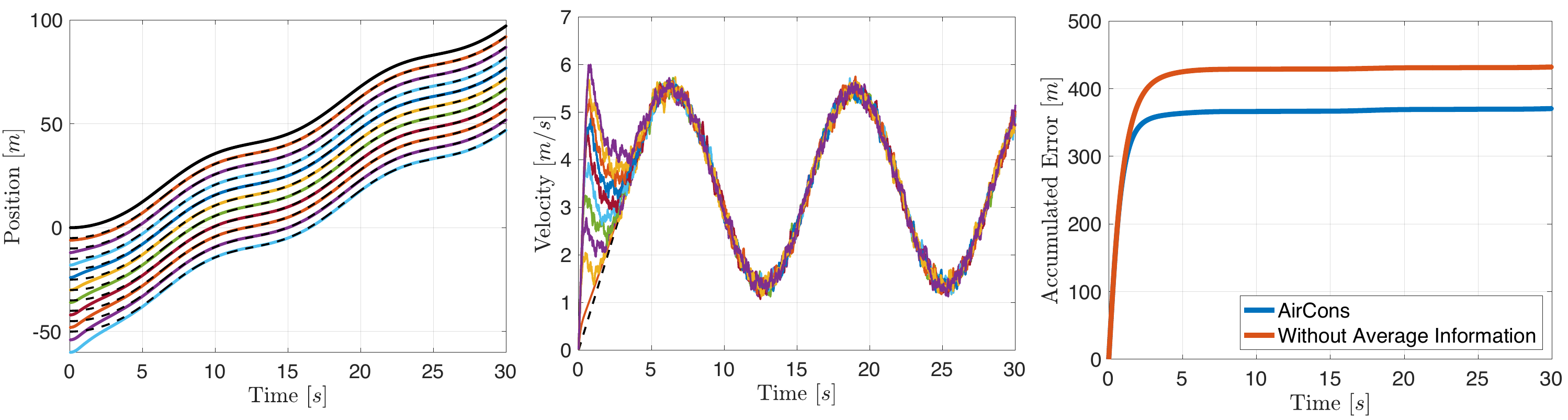}
\caption{Platooning control and the resultant performance of AirCons over time. (Left) Position trajectory of each AV in a platoon. Colored and dotted lines show each AV's actual and desired position, respectively. (Center) Each AV's velocity over time. (Right) The comparison of accumulated position error between AirCons and the benchmark without average information.}
\vspace{4pt}\label{Fig_platoon_position}
\end{figure*}

This section provides simulation results to verify the effectiveness of AirCons in terms of stability and control performance. As a benchmark, we consider   
a leader-predecessor following control mechanism where a leader and predecessor AVs' positions are used to control each AV's accelerator. Simulation parameters are set as follows unless specified otherwise. 
The number of AVs in a platoon $N$ is $10$.
Each AV's maximum transmission power $P$ is $23$ (dBm).
The total bandwidth $B$ is $20$ (MHz). A noise spectral density is $-174$ (dBm/Hz). The path-loss exponent $\alpha$ is set as $4$. The target inter-vehicle distance is  $d=5$ (m). 

To check string stability, we model a leader vehicle's turbulence as follows. At the initial stage, the leader vehicle increases its speed with a constant acceleration of $1$ (m/$\mathrm{s}^2$). After $5$ seconds, the leader vehicle repeatedly accelerates and decelerates by setting its accelerator as $10\sin{\frac{t}{2}}$ (m/$\mathrm{s}^2$). As recalled in Assumption \ref{Assumption1}, the leader AV's position information is periodically broadcast. We set the broadcasting interval as $10$ (ms). Besides, each AV can scan the predecessor AV's inter-distance using RADAR sensors with the interval of $10$ (ms). The number of iterations for a consensus is $K=6$, and the resultant delay of the average position estimate is $6\times 0.915=5.49$ (ms), where $915$ ($\mu$s) is the coherence time specified in Remark \ref{Coherence_Time_Bandwidth}.

Fig. \ref{Fig_platoon_position} represents the performance of AirCons with several interesting observations. First, the left-side figure shows each AV's trajectory (solid curves) with $10$ different colors representing each of them. It is shown that every AV follows the desired path (dotted curves) except the initial phase. Each AV's velocity over time is plotted in the middle figure. It is shown that each AV's velocity follows the leader AV's one once the platoon becomes stabilized, thereby satisfying string stability under the current turbulence setting. The right-side figure represents the accumulated errors of AirCons and benchmarks, both of which meet string stability. On the other hand, the average position information helps reach a consensus faster than the benchmark, resulting in  $14.22\%$ reduction of the accumulated error.

\section{Conclusion}\label{Section: Conclusion}
This work has proposed AirCons, a joint communication-and-control design for distributed VP. The main goal of AirCons is to estimate the average position of neighbor AVs instead of individual position information, which is verified via extensive simulations to provide a significant control gain. Exploiting a wireless signal's superposition and broadcasting properties enable multiple AVs' received signals to reach the consensus in the air, which is close to the desired average position. Besides, AirCons do not require establishing individual V2V links, reducing the significant burden of frequent training of many channels. To make the proposed AirCons more practical, we consider several interesting directions for future work, e.g., the extensions to radio resource management and complex driving scenarios. 
 
\bibliographystyle{IEEEtran}
\bibliography{ref}

\begin{thebibliography}{10}
\providecommand{\url}[1]{#1}
\csname url@samestyle\endcsname
\providecommand{\newblock}{\relax}
\providecommand{\bibinfo}[2]{#2}
\providecommand{\BIBentrySTDinterwordspacing}{\spaceskip=0pt\relax}
\providecommand{\BIBentryALTinterwordstretchfactor}{4}
\providecommand{\BIBentryALTinterwordspacing}{\spaceskip=\fontdimen2\font plus
\BIBentryALTinterwordstretchfactor\fontdimen3\font minus
  \fontdimen4\font\relax}
\providecommand{\BIBforeignlanguage}[2]{{%
\expandafter\ifx\csname l@#1\endcsname\relax
\typeout{** WARNING: IEEEtran.bst: No hyphenation pattern has been}%
\typeout{** loaded for the language `#1'. Using the pattern for}%
\typeout{** the default language instead.}%
\else
\language=\csname l@#1\endcsname
\fi
#2}}
\providecommand{\BIBdecl}{\relax}
\BIBdecl

\bibitem{ACC_1}
P.~Ioannou and C.~Chien, ``Autonomous intelligent cruise control,'' \emph{IEEE
  Trans. Veh. Technol.}, vol.~42, no.~4, pp. 657--672, 1993.

\bibitem{energy_efficient_2}
A.~A. Hussein and H.~A. Rakha, ``Vehicle platooning impact on drag coefficients
  and energy/fuel saving implications,'' \emph{IEEE Trans. Veh. Technol.},
  vol.~71, no.~2, pp. 1199--1208, 2022.

\bibitem{distibuted_consensus}
H.~Seo, J.~Park, M.~Bennis \emph{et~al.}, ``Communication and consensus
  co-design for distributed, low-latency, and reliable wireless systems,''
  \emph{IEEE Internet Things J.}, vol.~8, no.~1, pp. 129--143, 2021.

\bibitem{String_stability1}
G.~J.~L. Naus, R.~P.~A. Vugts, J.~Ploeg \emph{et~al.}, ``String-stable {CACC}
  design and experimental validation: {A} frequency-domain approach,''
  \emph{IEEE Trans. Veh. Technol.}, vol.~59, no.~9, pp. 4268--4279, 2010.

\bibitem{NLoS_V2V}
M.~Noor-A-Rahim, G.~G. M.~N. Ali, Y.~L. Guan \emph{et~al.}, ``Broadcast
  performance analysis and improvements of the lte-v2v autonomous mode at road
  intersection,'' \emph{IEEE Transactions on Vehicular Technology}, vol.~68,
  no.~10, pp. 9359--9369, 2019.

\bibitem{impact_of_time_delay}
X.~Liu, A.~Goldsmith, S.~Mahal \emph{et~al.}, ``Effects of communication delay
  on string stability in vehicle platoons,'' in \emph{Proc. ITSC 2001. 2001
  IEEE Intell. Transp. Syst. (Cat. No.01TH8585)}, 2001, pp. 625--630.

\bibitem{VP_delay2}
F.~Ma, J.~Wang, S.~Zhu \emph{et~al.}, ``Distributed control of cooperative
  vehicular platoon with nonideal communication condition,'' \emph{IEEE Trans.
  Veh. Technol.}, vol.~69, no.~8, pp. 8207--8220, 2020.

\bibitem{pilot_contamination}
O.~Elijah, C.~Y. Leow, T.~A. Rahman \emph{et~al.}, ``A comprehensive survey of
  pilot contamination in massive {MIMO—5G} system,'' \emph{IEEE Commun.
  Surveys Tuts.}, vol.~18, no.~2, pp. 905--923, 2016.

\bibitem{AFC_imperfect_CSI}
H.~Jung and S.-W. Ko, ``Performance analysis of {UAV}-enabled over-the-air
  computation under imperfect channel estimation,'' \emph{IEEE Wireless Commun.
  Lett.}, vol.~11, no.~3, pp. 438--442, 2022.

\bibitem{overtheair1001}
G.~Zhu, J.~Xu, K.~Huang \emph{et~al.}, ``Over-the-{Air} computing for wireless
  data aggregation in massive iot,'' \emph{IEEE Wireless Commun.}, vol.~28,
  no.~4, pp. pp.57--65, 2021.

\bibitem{1003}
S.-W. Ko, H.~Chae, K.~Han \emph{et~al.}, ``{V2X}-based vehicular positioning:
  Opportunities, challenges, and future directions,'' \emph{IEEE Wireless
  Commun.}, vol.~28, no.~2, pp. 144--151, 2021.

\bibitem{typical_controller1004}
S.~Santini, A.~Salvi, A.~S. Valente \emph{et~al.}, ``A consensus-based approach
  for platooning with intervehicular commun. and its validation in realistic
  scenarios,'' \emph{IEEE Trans. Veh. Technol.}, vol.~66, no.~3, pp.
  1985--1999, 2017.

\bibitem{typical_controller1005}
------, ``Platooning maneuvers in vehicular networks: A distributed and
  consensus-based approach,'' \emph{IEEE Trans. Intell. Veh.}, vol.~4, no.~1,
  pp. 59--72, 2019.

\bibitem{Numerology}
C.~Campolo, A.~Molinaro, F.~Romeo \emph{et~al.}, ``5{G} {NR} {V2X}: On the
  impact of a flexible numerology on the autonomous sidelink mode,'' in
  \emph{2019 IEEE 2nd 5G World Forum (5GWF)}, 2019, pp. 102--107.

\bibitem{3gpp.37.885}
3GPP, ``Study on evaluation methodology of new vehicle-to-everything ({V2X})
  use cases for {LTE} and {NR},'' {3rd Generation Partnership Project (3GPP)},
  Technical Report ({TR}) 37.885, 06 2019, version 15.3.0.

\bibitem{Average_Consensus_Controller1006}
F.~Molinari, S.~Stanczak, and J.~Raisch, ``Exploiting the superposition
  property of wireless communication for average consensus problems in
  multi-agent systems,'' in \emph{Proc. 2018 European Control Conference
  (ECC)}, 2018, pp. 1766--1772.

\end{thebibliography}

\appendix
\subsection{Proof of Theorem \ref{Theorem1}}\label{Appendix_A} With noise being neglected and 
$k>1$, \eqref{Eq:Updating_x} is rewritten as 
\begin{align}\label{Eq: revised_updating_x}
    x_m[k+1]=&(1-\rho) x_{m}[k] \nonumber\\&+\rho \left(\frac{\mathsf{Re}\{\sum_{\ell\in \mathbb{T}_{n}/\{m\}} h_{m,\ell}[k]I_{\ell}^*[k] x_{\ell}[k]\}}{\mathsf{Re}\left\{\sum_{\ell\in \mathbb{T}_{n}/\{m\}} h_{m,\ell}  [k] I_\ell^*[k]\right\}}\right)\nonumber\\
    \overset{(a)}{=}&(1-\rho) x_{m}[k]+\rho \left(\frac{\sum_{\ell\in \mathbb{T}_n/\{m\}} |a_{m,\ell}[k]| x_{\ell}[k]}{\sum_{\ell\in \mathbb{T}_n/\{m\}} |a_{m,\ell}[k]| }\right),
\end{align}
where $a_{m,\ell}[k]=\mathsf{Re}\{h_{m,\ell}[k]\}$ are specified in \eqref{a_k}, and (a) follows from the fact that $I_{\ell}^*[k]$ and $x_\ell[k]$ are pure real term. Denote $\mathbf{x}[k]=\left[x_1[k],x_2[k],...,x_{S}[k]\right]^T$. Then, \eqref{Eq: revised_updating_x} can be expressed as $\mathbf{x}_m[k+1]=\mathbf{b}_m[k]^T \mathbf{x}[k]$.
\begin{align}
    \mathbf{x}[k+1]=\mathbf{B}[k]\mathbf{x}[k],
\end{align}
where $\mathbf{B}[k]=\left[\mathbf{b}_1[k],\cdots, \mathbf{b}_{S}[k]\right]^T$. 
Given $\rho\geq 0$, all elements of $\mathbf{b}_m[k]$ are strictly positive and their sum is always one, confirming that $\mathbf{B}[k]$ is a row-stochastic matrix. It is proved in \cite{Average_Consensus_Controller1006} that $\mathbf{x}[k]$ converges to the weighted  average of the initial value $\mathbf{x}[0]$, namely, 
\begin{align}
\mathbf{w}^T\mathbf{x}[1]=\sum_{m\in\mathbb{S}_n\cup{\{n\}}}w_m s_m=\frac{\sqrt{\rho P}}{L}\sum_{m\in\mathbb{T}_n}w_m \alpha_m, 
\end{align}
which completes the proof.

\subsection{Proof of Proposition \ref{Proposition1}} \label{Appendix_B}
Channels are  time-invariant. That is $\eta_{m,\ell}=\sqrt{\frac{2}{\pi}}\frac{1}{|m-\ell|}$. Then matrix $\boldsymbol{U}$ is written as follows:
\begin{align}\label{Eq: Centro-symmetric}
    &\boldsymbol{U}=\nonumber\\
  &\begin{pmatrix}
        1-\rho & \frac{\frac{1}{d|1-2|}}{{\sum_{m \in \mathbb{T}_{n}/\{1\}}\frac{1}{d|1-m|}} } &\dots &\frac{\frac{1}{d|1-S|}}{\sum_{m \in \mathbb{T}_{n}/\{1\}}\frac{1}{d|1-m|}}\\\frac{\frac{1}{d|2-1|}}{{\sum_{m \in \mathbb{T}_{n}/\{2\}}\frac{1}{d|2-m|}} }& 1-\rho & \dots & \frac{\frac{1}{d|2-S|}}{{\sum_{m \in \mathbb{T}_{n}/\{2\}}\frac{1}{d|2-m|}} }\\ \vdots & \vdots &\ddots &\vdots\\ \frac{\frac{1}{d|S-1|}}{{\sum_{m \in \mathbb{T}_{n}/\{S\}}\frac{1}{d|S-m|}}}& \frac{\frac{1}{d|S-2|}}{{\sum_{m \in \mathbb{T}_{n}/\{S\}}\frac{1}{d|S-m|}}} &\dots &1-\rho
  \end{pmatrix}.
\end{align}
Since, $ \sum_{m \in \mathbb{T}_{n}/\{\ell\}}\frac{1}{d|\ell-m|}=\sum_{m \in \mathbb{T}_{n}/\{S+1-\ell\}}\frac{1}{d|S+1-\ell-m|}$, matrix element  $u_{m,\ell}=u_{(S+1-m,S+1-\ell)}$. Then matrix $\boldsymbol{U}$ is centro-symmetric matrix. Left eigen vector corresponding to eigen value $1$ is denoted as $\boldsymbol{v}$. Then,
\begin{align}\label{Eq: left_eigen}
    \boldsymbol{vU}=1\cdot \boldsymbol{v}.
\end{align}
Then, $\boldsymbol{v}(\boldsymbol{U}-\boldsymbol{I})=\boldsymbol{0}$, where $\boldsymbol{I}$ denotes identity matrix and $\boldsymbol{0}$ denotes matrix of size ($1\times S$) whose elements are zeros. 
It is obvious that left eigen vector $\boldsymbol{v}$ is symmetric. Since inter-distance of AVs are assumed to be strictly equal, we say $[\alpha_1,\cdots,\alpha_{S}]=[d ,2d,...Sd]$. Then, $\boldsymbol{v\alpha}=\frac{1}{2}\left(S+1\right)d$
Then, we conclude that $\boldsymbol{v\alpha}=\zeta_n$. Thus proof is completed.
\end{document}